\documentclass[a4paper,12pt]{article}
\usepackage{fullpage}
\usepackage[bookmarksnumbered]{hyperref}
\usepackage[english]{babel}
\usepackage[utf8]{inputenc}
\usepackage{mathtools}
\usepackage{verbatim}
\usepackage{amsmath}
\usepackage{amssymb}
\usepackage{amsthm}
\usepackage{lipsum}
\usepackage{bbold}
\usepackage[capitalize]{cleveref}
\title{On the Universality and Membership problems for quantum gates}
\author{Lorenzo Mattioli${}^{1,*}$, Adam Sawicki${}^{1,**}$}
\date{
${}^{1}$ Center for Theoretical Physics, Polish Academy of Sciences, \\
Al. Lotników 32/46, 02-668 Warsaw, Poland \\
${}^{*}$mattioli@cft.edu.pl, ${^{**}}$ a.sawicki@cft.edu.pl \\
[2ex]
\today
}

\theoremstyle{definition}
\newtheorem{definition}{Definition}
\theoremstyle{plain}
\newtheorem{lemma}{Lemma}
\theoremstyle{plain}
\newtheorem{theorem}{Theorem}
\theoremstyle{plain}

\theoremstyle{remark}
\newtheorem{remark}{Remark}
\begin{document}
\maketitle
\begin{abstract}
We study the \emph{Universality} and \emph{Membership Problems} for gate sets consisting of a finite number of quantum gates. Our approach relies on  the techniques from compact Lie groups theory.
We also introduce an auxiliary problem called \emph{Subgroup Universality Problem}, which helps in solving some instances of the Membership Problem, and can be of interest on its own.
The resulting theorems are mainly formulated in terms of centralizers and the adjoint representations of a given set of quantum gates.
\end{abstract}
\section{Introduction}
It is well known that to construct any logic operation a classical computer needs just a single $2$-bit gate (such as NAND), since we can compose such gate many times in different ways and build any other gate.
Due to nonabelian structure of the unitary group the situation is manifestly not the same for quantum computers. Compositions of one quantum gate can only generate an abelian subgroup. Moreover, the set of all unitary operations is uncountable. Thus using \emph{parametrized} gates (e.g. all the rotations about certain axes), we can, under some conditions on the choice of the axes, precisely build any other gate. On the other hand if we only have access to a finite set $\mathcal{S}$ of \emph{not-parametrized} gates instead, we can implement precisely only countable subset of unitary operations, which we denote by $<\mathcal{S}>$. If this countable set is dense in the set of unitary operations, we call $\mathcal{S}$ a universal set of gates. Alternatively one can say that $\mathcal{S}$ is universal iff by composing gates from $\mathcal{S}$  can can approximate arbitrarily well any unitary operations.
A celebrated result in quantum computing states that a universal gate-set for one qubit together with one additional $2$-qubit entangling gate is a universal gate-set for any number of qubits \cite{Brylinski,Nielsen} (see also \cite{Oszmaniec1} for fermionic systems).

It is not difficult to choose a universal gate-set, since almost any gate-set of fixed cardinality bigger or equal two is universal (i.e. universal gate-sets form a Zariski open set \cite{Kuranishi}) and particular universal gate-sets are known \cite{Nielsen}. Despite this, recently there has been some activity regarding the \emph{Universality Problem}, which asks whether a given gate-set is universal. An algorithm for universality checking has been proposed \cite{Adam,Sawicki2, Sawicki3}. There are a few natural generalizations of the universality problem that we study in this paper. First of all, one can consider the \emph{Membership Problem}, which asks whether we can approximate arbitrarily well just certain gates (no matter whether the gate-set is universal or not).
Next, since there are proposed physical implementation of quantum computers which rely on quantum optics \cite{Exp1,Exp2,Reck,Sawicki1} and exploit $d \geq 2$ modes of light, we can generalize the problems to qudits (\cite{Fermion} provides another motivation in the context of fermionic linear optics).
Finally, since the unitary group is a compact Lie group, we can generalize the problems to any other compact Lie group $K$. 
This generalization can be of some relevance in the context of VQAs, where the phenomenon of \emph{barren plateaus} can in principle be avoided at the cost of reducing the expressibility of the ansatz (i.e. restricting the reachable set of gates to a subgroup) \cite{BarrenPlateaus}.
Our generalization also allows to reformulate the problems at the level of the Lie algebra $\mathfrak{k}$ of $K$, in the hope that solving these new problems help for solving the original ones. 

Following this line of reasoning we consider a subalgebra $\mathfrak{g}$ of $\mathfrak{k}$ generated by a subset $\mathcal{X}$ of $\mathfrak{k}$,
but we do \emph{not} want to generate $\mathfrak{g}$ explicitly, i.e. by taking linear combinations of nested commutators of elements in $\mathfrak{k}$.
Instead, we follow the approach taken in \cite{Zeier1,Zoltan1,Zoltan2,Adam} and in Section \ref{sec:Liealgebra} we review the results considering the universality and the membership problems for compact Lie algebras. In our notation the membership problem concerns checking if a subset $\mathcal{Y}$ of $\mathfrak{k}$ (possibly just one element of $\mathfrak{k}$) satisfies $\mathcal{Y} \subseteq \mathfrak{g}$. The problem has been solved in \cite{Zoltan1,Zoltan2} and we review it in \cref{membershipproblemalgebra}.

Our main goal is to consider analogous problems for finite sets of gates. In our setting we are given a subset $S\subset K$, where $K=SU(d)$ and we define the Lie subgroup $H\subset K$ as the closure of the set of words whose alphabet is $S$. The Universality Problem for Lie groups asks whether $H = K$. The problem was studied in \cite{Adam,Sawicki2}. Here we provide an alternative proof of Lemma 4.8 of \cite{Adam} that relies on properties of an auxiliary space defined in \cref{defofa}. We believe that  the techniques we developed can be of some interest on its own and they are also crucial for the membership problem.

Next we consider the membership problem for Lie groups, i.e. for a given subset  $T$ of $K$ (possibly just one element of $K$) we ask whether $T \subseteq H$. In order to deal with this we introduce some additional structure and an auxiliary problem which we call \emph{Subgroup Universality Problem}, that we describe in the next paragraph.

In the subgroup universality problem we consider the subalgebra $\mathfrak{g}$ of $\mathfrak{k}$ generated by a subset $\mathcal{X}\subset\mathfrak{k}$ and the unique connected Lie subgroup $G\subset K$ whose Lie algebra is $\mathfrak{g}$.
Explicitly, G is the set of words whose alphabet is either the exponential of $\mathfrak{g}$ or just the one-parameter subgroups corresponding to $\mathcal{X}$ (\cite{Parameter1,Parameter2}).
We finally define $S$ as the exponential of $\mathcal{X}$ and $H$ as the closure of the set of words whose alphabet is $S$.
The Subgroup Universality Problem asks if $H = G$ and we solve it for $\mathfrak{g}$ simple in \cref{subgroupuniversalityproblem}.
Having \cref{subgroupuniversalityproblem} we go back to the Membership Problem for Lie groups and distinguish a few cases in which we can solve it (\cref{membershipproblemgroup} and \cref{membershipproblemgroup2}). We also comment on the limitations of this approach.

All of our theorems have the following form: the answer to a given problem is 'yes' \emph{if and only if} some centralizers involving the adjoint representation of the Lie algebra (or Lie group) are equal \emph{and} some additional condition holds.
In some sense, the goal of this paper is \emph{not} solving our problems completely (which in general is a hopeless task, as shown for instance in \cite{Hardness}), but showing how far we can go with statements of that form.
The choice of the \emph{adjoint representation} of a Lie algebra or Lie group is natural one as every Lie algebra and every Lie group has the adjoint representation.
The additional condition are expressed directly in terms of the Lie algebra or the Lie group and is needed since the adjoint representation does not encode full information about the Lie algebra or a Lie group.
\section{Universality and membership problems for Lie algebras}\label{sec:Liealgebra}
In this section we introduce some basic concepts that will play a crucial role throughout the paper and review the results concerning the universality and membership problems for Lie algebras \cite{Zeier1,Zoltan1,Zoltan2,Adam}. We assume the reader is familiar with the basics of Lie algebras, and refer to the appendix for the technicalities concerning \emph{reductive} and \emph{compact} Lie algebras. Unless otherwise stated, we will be using the Killing form of a Lie algebra as a natural inner product.
\begin{definition}\label{subalgebragenerated}
Let $\mathfrak{k}$ be a Lie algebra and let $\mathcal{X}$ be a subset of $\mathfrak{k}$. Then
$\left\langle \mathcal{X} \right\rangle$, the \emph{subalgebra of $\mathfrak{k}$ generated by $\mathcal{X}$}, is any of the following equivalent objects:
\begin{itemize}
\item the \emph{minimal} subalgebra of $\mathfrak{k}$ which contains $\mathcal{X}$,
\item the subalgebra of $\mathfrak{k}$ to which the following sequence of subspaces of $\mathfrak{k}$ converges:
\begin{equation*}
\begin{split}
{W}_{0}
& =
\text{Span} \left( \mathcal{X} \right),
\\
{W}_{i + 1}
& =
{W}_{i} + \left[ {W}_{i},{W}_{i} \right].
\end{split}
\end{equation*}
\end{itemize}
\end{definition}
\noindent Another important observation is that $\left\langle \cdot \right\rangle$ commutes with Lie homomorphisms, i.e. linear maps that preserves Lie bracket, i.e the commutator.

\begin{lemma}\label{homomorphismgenerators}
Let $\mathfrak{k}$ and $\widetilde{\mathfrak{k}}$ be Lie algebras. Let $f : \mathfrak{k} \mapsto \widetilde{\mathfrak{k}}$ be a Lie homomorphism. Let $\mathcal{X}$ be a subset of $\mathfrak{k}$.

\begin{equation*}
f \left( \left\langle \mathcal{X} \right\rangle \right)
=\left\langle f \left( \mathcal{X} \right) \right\rangle.
\end{equation*}
\end{lemma}
\noindent We continue with the definition of the notion of a {\it centralizer}.
\begin{definition}\label{centralizeralgebra}
Let $\mathfrak{k}$ be a Lie algebra. Let $\mathfrak{g}$ be a subalgebra of $\mathfrak{k}$ and let $\mathcal{X}$ be a subset of $\mathfrak{k}$. The \emph{centralizer of $\mathcal{X}$ in $\mathfrak{g}$} is the following subalgebra of $\mathfrak{g}$:
\begin{equation*}
{C}_{\mathfrak{g}} \left( \mathcal{X} \right)
\coloneqq
\left\{ x \in g \mid \forall \ y \in \mathcal{X} \left[ x,y \right] = 0 \right\}.
\end{equation*}
\end{definition}
\noindent Having a subset $\mathcal{X}\subset \mathfrak{g}$ there are quantities depending on a generated subalgebra, $\left\langle \mathcal{X} \right\rangle$, that can be calculated without generating $\left\langle \mathcal{X} \right\rangle$. As an immediate consequence of Definition \ref{centralizeralgebra} we get the following lemma: 
\begin{lemma}\label{centralizeralgebragenerators}
Let $\mathfrak{k}$ be a Lie algebra and $\mathfrak{g}$ be a subalgebra of $\mathfrak{k}$. Let $\mathcal{X}$ be a subset of $\mathfrak{k}$.
\begin{equation*}
{C}_{\mathfrak{g}} \left( \left\langle \mathcal{X} \right\rangle \right)
={C}_{\mathfrak{g}} \left( \mathcal{X} \right).
\end{equation*}
\end{lemma}
\noindent For \emph{compact} Lie algebras we introduce projection operators associated with the set $\mathcal{X}$.

\begin{definition}\label{projector} Let $\mathfrak{k}$ be a \emph{compact} Lie algebra and $\mathcal{X}$ be a subset of $\mathfrak{k}$. ${P}_{\mathcal{X}}$ is the projector such that $\text{Im} \left( {P}_{\mathcal{X}} \right) = {C}_{\mathfrak{k}} \left( \mathcal{X} \right)$ and $\text{Ker} \left( {P}_{\mathcal{X}} \right) = {\mathfrak{k}}^{\prime} \cap {{C}_{\mathfrak{k}} \left( \mathcal{X} \right)}^{\perp}$.
\end{definition}

\noindent One can see (cf. \cref{compactdecomposition}) that ${P}_{\mathcal{X}}$ is well defined, namely that $\mathfrak{k} = \text{Im} \left( {P}_{\mathcal{X}} \right) \oplus \text{Ker} \left( {P}_{\mathcal{X}} \right)$. Similarly as in \cref{centralizeralgebragenerators} we have 
\begin{lemma}\label{projectorgenerators}
Let $\mathfrak{k}$ be a \emph{compact} Lie algebra and $\mathcal{X}$ be a subset of $\mathfrak{k}$. Then ${P}_{\left\langle \mathcal{X} \right\rangle}={P}_{\mathcal{X}}.$
\end{lemma}
\noindent Once we know a decomposition of a generated subalgebra, it is natural to ask how to generate the summands in the decomposition independently. The following theorem provides the answer for \emph{compact} Lie algebras (see \cref{reductivesubalgebra}), and its first equality happens to be yet another instance of our ``golden rule''.

\begin{lemma}\label{subgenerators}
Let $\mathfrak{k}$ be a \emph{compact} Lie algebra. 
Let $\mathcal{X}$ be a subset of $\mathfrak{k}$. 
Let $\mathfrak{g} \coloneqq \left\langle \mathcal{X} \right\rangle$.
\begin{equation*}
\begin{split}
{C}_{\mathfrak{g}} \left( \mathfrak{g} \right)
& =
\mathrm{Span} \left( {P}_{\mathcal{X}} \left( \mathcal{X} \right) \right),
\\
{\mathfrak{g}}^{\prime}
& =
\left\langle \left( \mathbb{1} - {P}_{\mathcal{X}} \right) \left( \mathcal{X} \right) \right\rangle.
\end{split}
\end{equation*}
\end{lemma}
%

\noindent We are finally ready to recall the results of \cite{Zoltan1,Zoltan2} and \cite{Adam}. The two results that are important from our perspective are the universality and membership problems for Lie algebras. In the first one we are given a finite\footnote{This assumption is obvious in real life, but irrelevant in the following theorem.} set of hamiltonians\footnote{Formally a subset of $\mathfrak{u} \left( d \right)$ after multiplication by the imaginary unit.} $\mathcal{X}$, and we ask if (and how) we can decide whether $\left\langle \mathcal{X} \right\rangle = \mathfrak{u} \left( d \right)$ without generating $\left\langle \mathcal{X} \right\rangle$. The general solution of this problem for any \emph{compact} Lie algebra $\mathfrak{k}$ is given in the following theorem \cite{Zoltan1,Zoltan2}.
\begin{theorem}\label{universalityproblemalgebra}
Let $\mathfrak{k}$ be a \emph{compact} Lie algebra. Let $\mathcal{X}$ be a subset of $\mathfrak{k}$ and $\mathfrak{g} \coloneqq \left\langle \mathcal{X} \right\rangle$.
\begin{equation*}
\mathfrak{g} = \mathfrak{k}
\iff
{C}_{\mathfrak{gl \left( k \right)}} \left( {\mathrm{ad}}_{\mathcal{X}} \right) = {C}_{\mathfrak{gl \left( k \right)}} \left( {\mathrm{ad}}_{\mathfrak{k}} \right)\,\, \mathrm{ and }\,\, \dim \left( \mathrm{Span} \left( {P}_{\mathfrak{k}} \left( \mathcal{X} \right) \right) \right) = \dim \left( {C}_{\mathfrak{k}} \left( \mathfrak{k} \right) \right).
\end{equation*}
\end{theorem}

\noindent For the membership problem, that is a natural generalization of the universality problem, we are given two finite sets of hamiltonians ${\mathcal{X}}_{1}$ and $\mathcal{Y}$, and we ask if (and how) we can decide whether $\mathcal{Y} \subset \left\langle {\mathcal{X}}_{1} \right\rangle$ without generating $\left\langle {\mathcal{X}}_{1} \right\rangle$.
Since $\left\langle {\mathcal{X}}_{1} \cup \mathcal{Y} \right\rangle$ is the minimal subalgebra which contains ${\mathcal{X}}_{1} \cup \mathcal{Y}$ and, at the same time, a subalgebra which contains ${\mathcal{X}}_{1}$, it follows that $\left\langle {\mathcal{X}}_{1} \right\rangle \subseteq \left\langle {\mathcal{X}}_{1} \cup \mathcal{Y} \right\rangle$.
Therefore we reformulate the problem and ask if we can decide whether $\left\langle {\mathcal{X}}_{1} \right\rangle = \left\langle {\mathcal{X}}_{1} \cup \mathcal{Y} \right\rangle$ without generating $\left\langle {\mathcal{X}}_{1} \right\rangle$ and $\left\langle {\mathcal{X}}_{1} \cup \mathcal{Y} \right\rangle$.
We generalize the problem solved in \cite{Zoltan1,Zoltan2} to any \emph{compact} Lie algebra $\mathfrak{k}$ and answer our question in the affirmative by the following theorem.
\begin{theorem}\label{membershipproblemalgebra}
Let $\mathfrak{k}$ be a \emph{compact} Lie algebra, ${\mathcal{X}}_{1}$ and $\mathcal{Y}$ be subsets of $\mathfrak{k}$ and let ${\mathcal{X}}_{2} \coloneqq {\mathcal{X}}_{1} \cup \mathcal{Y}$. Let $\mathfrak{{g}_{1}} \coloneqq \left\langle {\mathcal{X}}_{1} \right\rangle$ and $\mathfrak{{g}_{2}} \coloneqq \left\langle {\mathcal{X}}_{2} \right\rangle$.
\begin{equation*}
\mathfrak{{g}_{1}} = \mathfrak{{g}_{2}}
\iff
{C}_{\mathfrak{gl \left( k \right)}} \left( {\mathrm{ad}}_{{\mathcal{X}}_{1}} \right) = {C}_{\mathfrak{gl \left( k \right)}} \left( {\mathrm{ad}}_{{\mathcal{X}}_{2}} \right) \text{ and } \dim \left( \text{Span} \left( {P}_{{\mathcal{X}}_{1}} \left( {\mathcal{X}}_{1} \right) \right) \right) = \dim \left( \text{Span} \left( {P}_{{\mathcal{X}}_{1}} \left( {\mathcal{X}}_{2} \right) \right) \right).
\end{equation*}
\end{theorem}

\begin{remark}\label{membershipremark}
\cref{membershipproblemalgebra} is still valid if we perform the substitution ${P}_{{\mathcal{X}}_{1}} \to {P}_{{\mathcal{X}}_{2}}$ in the condition
\begin{equation*}
\dim \left( \text{Span} \left( {P}_{{\mathcal{X}}_{1}} \left( {\mathcal{X}}_{1} \right) \right) \right) = \dim \left( \text{Span} \left( {P}_{{\mathcal{X}}_{1}} \left( {\mathcal{X}}_{2} \right) \right) \right).
\end{equation*}
\end{remark}
\section{Universality and membership problems for quantum gates}
In this section we discuss (subgroup) universality and membership problems in the context of quantum gates rather than Hamiltonians. Having a set of Hamiltonians, $\mathcal{X}$, we can consider associated quantum gates, i.e. the set of group elements $S=\{\exp(X)|\,X\in \mathcal{X}\}$. On the other hand we can also consider one parameter subgroups $\tilde{S}=\{\exp(tX)|\,X\in \mathcal{X},\,t\in\mathbb{R}\}$. The scenario with $S$ corresponds to situation when we can apply our Hamiltonians only for fixed amount of time and the scenario with $\tilde{S}$ places no restriction on the time of Hamiltonians application. It is natural to ask under which conditions the groups generated by $\mathcal{S}$ and $\tilde{S}$ coincide. Furthermore we will also investigate what happens when we extend the set of Hamiltonians $\mathcal{X}$ by an additional element $Y$ which is not present in  $\langle \mathcal{X}\rangle$. The condition $Y\notin\langle \mathcal{X}\rangle$ can be checked by \cref{membershipproblemalgebra}. This will lead us to the solution of a certain instance of the membership problem for quantum gates. 

Let us start with fixing notation and recalling the basic facts.
\begin{definition}\label{words}
Let $K$ be a group and $S$ be a subset of $K$. We define the set of words of length $n$ constructed from elements of $S$ by ${S}^{n}
\coloneqq\left\{ {g}_{1} \cdots {g}_{n} \mid {g}_{k} \in S \right\}$. By $\left\langle S \right\rangle$ we denote a (possibly infinite) set of words of all lengths, $\left\langle S \right\rangle \coloneqq\bigcup_{n = 0}^{\infty} {S}^{n}$
\end{definition}
\noindent Deciding what is $\left\langle S \right\rangle$ requires techniques that do not relay on finding $S^n$ for every $n$ as this can be an infinite process. As was shown in \cite{Adam} the crucial role is played by the so-called adjoint representation $\mathrm{Ad}:K\rightarrow GL(\mathfrak{k})$ which for matrix Lie groups is defined by the conjugation $\mathrm{Ad}_g(X):=gXg^{-1}$ (see appendix and \cite{Adam} for more details regarding relevant properties of the adjoint representation ). The following two lemmas are fundamental properties that are useful from our purposes.
\begin{lemma}\label{centralizergroupgenerators}
Let $K$ be a Lie group and $\mathfrak{k}$ be the Lie algebra of $K$. Let $S$ be a subset of $K$. Then
\begin{equation*}
{C}_{\mathfrak{gl \left( k \right)}} \left( {\mathrm{Ad}}_{S} \right)
=
{C}_{\mathfrak{gl \left( k \right)}} \left( {\mathrm{Ad}}_{\left\langle S \right\rangle} \right).
\end{equation*}
\end{lemma}
\begin{proof}
 Obviously $S\subseteq\left\langle S \right\rangle$ and hence ${\mathrm{Ad}}_{S}\subseteq {\mathrm{Ad}}_{\left\langle S \right\rangle}$ which implies ${C}_{\mathfrak{gl \left( k \right)}} \left( {\mathrm{Ad}}_{\left\langle S \right\rangle} \right)\subseteq{C}_{\mathfrak{gl \left( k \right)}} \left( {\mathrm{Ad}}_{S} \right)$.
%
%
%
On the other hand $\forall \ A \in {C}_{\mathfrak{gl \left( k \right)}} \left( {\mathrm{Ad}}_{S} \right)$ we have 
$\forall \ g,h \in S$ $\left[ A,{\mathrm{Ad}}_{gh} \right] = \left[ A,{\mathrm{Ad}}_{g} {\mathrm{Ad}}_{h} \right] = {\mathrm{Ad}}_{g} \left[ A,{\mathrm{Ad}}_{h} \right] + \left[ A,{\mathrm{Ad}}_{g} \right] {\mathrm{Ad}}_{h} = 0$. Thus
$A \in {C}_{\mathfrak{gl \left( k \right)}} \left( {\mathrm{Ad}}_{\left\langle S \right\rangle} \right)$ and ${C}_{\mathfrak{gl \left( k \right)}} \left( {\mathrm{Ad}}_{S} \right) \subseteq {C}_{\mathfrak{gl \left( k \right)}} \left( {\mathrm{Ad}}_{\left\langle S \right\rangle} \right)$

\end{proof}

\begin{lemma}\label{centralizerclosure}
Let $K$ be a Lie group, $\mathfrak{k}$ be the Lie algebra of $K$ and $U$ be a subset of $K$. Then
\begin{equation*}
{C}_{\mathfrak{gl \left( k \right)}} \left( {\mathrm{Ad}}_{U} \right)
=
{C}_{\mathfrak{gl \left( k \right)}} \left( {\mathrm{Ad}}_{\overline{U}} \right).
\end{equation*}
\end{lemma}

\begin{proof}
Obviously $U\subseteq\overline{U}$ hence ${\mathrm{Ad}}_{U}\subseteq{\mathrm{Ad}}_{\overline{U}}$ which implies ${C}_{\mathfrak{gl \left( k \right)}} \left( {\mathrm{Ad}}_{\overline{U}} \right)\subseteq{C}_{\mathfrak{gl \left( k \right)}} \left( {\mathrm{Ad}}_{U} \right)$. On the other hand $\forall \ A \in {C}_{\mathfrak{gl \left( k \right)}} \left( {\mathrm{Ad}}_{U} \right)$ we have $\left[ A,{\mathrm{Ad}}_{g} \right] = 0$, $\forall \ g \in U$. Next, by \cref{parameter} $\forall \ g \in U$ $\forall \ t \in \mathbb{R}$ $\left[ {e}^{t A},{\mathrm{Ad}}_{g} \right] = 0$ which means
$\forall \ g \in U$ $\forall \ t \in \mathbb{R}$ ${e}^{t A} {\mathrm{Ad}}_{g} {e}^{- t A} {\mathrm{Ad}}_{{g}^{-1}} = \mathbb{1}$. Next,since the adjoint representation and the group commutator are continuous, $g \mapsto {e}^{t A} {\mathrm{Ad}}_{g} {e}^{- t A} {\mathrm{Ad}}_{{g}^{-1}}$ is continuous, and since a continuous function constant on a subset is constant on its closure, we get 
$\forall \ g \in \overline{U}$ $\forall \ t \in \mathbb{R}$ ${e}^{t A} {\mathrm{Ad}}_{g} {e}^{- t A} {\mathrm{Ad}}_{{g}^{-1}} = \mathbb{1}$. This in turn means
$\forall \ g \in \overline{U}$ $\forall \ t \in \mathbb{R}$ $\left[ {e}^{t A},{\mathrm{Ad}}_{g} \right] = 0$. Using \cref{parameter} we obtain $\forall \ g \in \overline{U}$ $\left[ A,{\mathrm{Ad}}_{g} \right] = 0$. Thus $A \in {C}_{\mathfrak{gl \left( k \right)}} \left( {\mathrm{Ad}}_{\overline{U}} \right)$ and ${C}_{\mathfrak{gl \left( k \right)}} \left( {\mathrm{Ad}}_{U} \right) \subseteq {C}_{\mathfrak{gl \left( k \right)}} \left( {\mathrm{Ad}}_{\overline{U}} \right)$.

\end{proof}

\noindent Let us remark here that our definition of $\left\langle S \right\rangle$ is non-standard, since usually in the literature $\left\langle S \right\rangle$ denotes the subgroup of $K$ generated by $S$, namely the minimal subgroup of $K$ which contains $S$. Although in general $\left\langle S \right\rangle$ is not a subgroup of $K$, the following theorem shows that it is `close' to be such in the case that we are dealing with.
\begin{theorem}\label{groupgenerated}
Let $K$ be a \emph{compact} Lie group and $S$ be a subset of $K$. Then $\overline{\left\langle S \right\rangle}$ is a compact Lie subgroup of K.

\end{theorem}
\begin{proof}
Since it is defined as a closure, $\overline{\left\langle S \right\rangle}$ is a closed subset of $K$. 
Since a closed subset of a compact set is also compact, $\overline{\left\langle S \right\rangle}$ is a compact subset of $K$. 
Since $K$ is compact, one can prove that $\overline{\left\langle S \right\rangle}$ is a closed subgroup of $K$ (see \cite{Adam}), and since a closed subgroup of a Lie group is a Lie subgroup (Cartan's theorem), $\overline{\left\langle S \right\rangle}$ is a Lie subgroup of $K$.
\end{proof}
Similarly to definition \cref{centralizeralgebra} we define a centralizer on the level of  a group.
\begin{definition}\label{centralizergroup}
Let $K$ be a group, $G$ be a subgroup of $K$ and $S$ be a subset of $K$. The \emph{centralizer of $S$ in $G$} is the following subgroup of $G$:
\begin{equation*}
{C}_{G} \left( S \right)
\coloneqq
\left\{ g \in G \mid \forall \ h \in S \ g h = h g \right\}.
\end{equation*}
\end{definition}
\noindent The universality and membership problems require also the notion of distance on a group. In this paper we will use a metric associated with the Frobenius norm of a matrix $A$, $\left\| A \right\| \coloneqq \sqrt{\mathrm{tr} \left( {A}^{\dagger} A \right)}$. Since the component of the identity is always generated by any neighborhood of the identity, we make a certain choice of this neighborhood (actually the neighborhood of the group centralizer):
\begin{definition}\label{balls}
Let $K$ be a subgroup of $U \left( d \right)$.
\begin{equation*}
\begin{split}
\quad & {\mathcal{B}}_{K} \left( c \right)
\coloneqq
\left\{ g \in K \mid \left\| g - c \right\| < \frac{1}{\sqrt{2}} \right\},\, c \in K 
\\
& {\mathcal{B}}_{K}
\coloneqq
\bigcup_{c \in {C}_{K} \left( K \right)} {\mathcal{B}}_{K} \left( c \right)
=
\left\{ g \in K \mid \min_{c \in {C}_{K} \left( K \right)} \left\| g - c \right\| < \frac{1}{\sqrt{2}} \right\}.
\end{split}
\end{equation*}
\end{definition}
In \cite{Adam} it was shown that gates from a set $S \subset SU \left( d \right)$ generate the group $K \coloneqq SU \left( d \right)$, i.e. $H \coloneqq \overline{\langle S\rangle}=SU(d)$ if and only if ${C}_{\mathfrak{gl \left( k \right)}} \left( {\mathrm{Ad}}_{S} \right) = \left\{ \lambda \mathbb{1} \mid \lambda \in \mathbb{R} \right\} \text{ and } H \cap {\mathcal{B}}_{K} \not\subseteq {C}_{K} \left( K \right)$. In fact the first condition, ${C}_{\mathfrak{gl \left( k \right)}} \left( {\mathrm{Ad}}_{S} \right) = \left\{ \lambda \mathbb{1} \mid \lambda \in \mathbb{R} \right\}$, is the necessary condition for the universality, i.e. if it is satisfied and $\langle S \rangle$ is infinite then $H=SU(d)$. The second condition, $H \cap {\mathcal{B}}_{K} \not\subseteq {C}_{K} \left( K \right)$ guarantees that $\langle S \rangle$ is infinite, provided that the necessary condition is satisfied. The assumption standing behind the proof of the second condition, $H \cap {\mathcal{B}}_{K} \not\subseteq {C}_{K} \left( K \right)$, that was used in the proof of this statement in \cite{Adam} was that the logarithm of a special unitary matrix is traceless skew - hermitian for all elements in $\mathcal{B}_{K}$. This turns out to be false, as we show in the appendix. Nevertheless, in the following we present an alternative proof that completely avoids the use of the logarithm but instead relies on an auxiliary space $\mathfrak{a}(H)$ that we define below. Let
\begin{equation}\label{xyg}
\begin{split}
& x \left( g \right)
\coloneqq
\frac{g - {g}^{- 1}}{2} - \mathrm{tr} \left( \frac{g - {g}^{- 1}}{2} \right) \frac{\mathbb{1}}{d}, \\
& y \left( g \right)
\coloneqq
\frac{g + {g}^{- 1}}{2 i} - \mathrm{tr} \left( \frac{g + {g}^{- 1}}{2 i} \right) \frac{\mathbb{1}}{d}.
\end{split}
\end{equation}
One easily checks that for $g \in SU \left( d \right)$ both $x \left( g \right)$, $y \left( g \right)$ belong to $\mathfrak{su} \left( d \right)$. Moreover, for any $g \in SU \left( d \right)$ we have the following decomposition
\begin{equation}
g
=
x \left( g \right) + i y \left( g \right) + \mathrm{tr} \left( g \right) \frac{\mathbb{1}}{d}.
\end{equation}

\begin{definition}\label{defofa}
Let $H$ be a subgroup of $K = SU \left( d \right)$. We define
\begin{equation*}
\mathfrak{a}(H) \coloneqq \mathrm{Span}_\mathbb{R}\left\{ x \left( g \right),y \left( g \right) \mid g \in H \cap {\mathcal{B}}_{K} \right\},
\end{equation*}
where $x \left( g \right)$ and $y \left( g \right)$ are given by \cref{xyg}.
\end{definition}
\noindent One can easily see that
\begin{equation}\label{inv1}
\begin{split}
& {\mathrm{Ad}}_{h} x \left( g \right)
=
x \left( h g {h}^{- 1} \right),\,g,h \in K
\\
& {\mathrm{Ad}}_{h} y \left( g \right)
=
y \left( h g {h}^{- 1} \right),\,g,h \in K.
\end{split}
\end{equation}
Combining (\ref{inv1}) with the invariance under conjugation of the Frobenius distance we conclude that $\mathfrak{a}(H)$ is an invariant subspace of $\mathrm{Ad}_H$. Moreover, we have the following lemma that was proved in \cite{Adam}
\begin{lemma}\label{abelian}
Let $K$ be a subgroup of $U \left( d \right)$ and $H$ be a subgroup of $K$.
\begin{equation*}
H \text{ is finite }
\implies
H \cap {\mathcal{B}}_{K} \text{ is abelian }.
\end{equation*}
\end{lemma}
As an immediate conclusion we get 
\begin{lemma}\label{aisabelian}
Let $K=SU(d)$ and  $\mathfrak{k}=Lie(K)$. Assume $H$ is finite. Then $\mathfrak{a}(H)$ is abelian subalgebra of $\mathfrak{k}$ that is invariant under $\mathrm{Ad}_H$.
\end{lemma}
We next formulate and prove two important lemmas that will be used in the proofs of our main results.
\begin{lemma}\label{invariantsubspace}
Let $K$ be a compact connected Lie group and  $\mathfrak{k}=Lie(K)$. Let $\mathfrak{g}$ be a subalgebra of $\mathfrak{k}$ and $H$ be a subgroup of $K$. Assume that ${C}_{\mathfrak{gl \left( k \right)}} \left( {\mathrm{Ad}}_{H} \right) = {C}_{\mathfrak{gl \left( k \right)}} \left( {\mathrm{ad}}_{\mathfrak{g}} \right)$. Then
\begin{equation*}
 W \text{ invariant subspace of } {\mathrm{Ad}}_{H} \implies \left[ \mathfrak{g},W \right] \subseteq W.
\end{equation*}
\end{lemma}
\begin{proof}
By \cref{compactgroup} there exists a definite symmetric bilinear form $B$ such that $\forall \ g \in G$ $\forall \ x,y \in \mathfrak{g}$ $B \left( {\mathrm{Ad}}_{g} x,{\mathrm{Ad}}_{g} y \right) = B \left( x,y \right)$. 
$\forall \ W$ invariant subspace of ${\mathrm{Ad}}_{H}$
\begin{equation*}
B \left( W,{\mathrm{Ad}}_{H} {W}^{\perp} \right)
=
B \left( {\mathrm{Ad}}_{H} W,{\mathrm{Ad}}_{H} {W}^{\perp} \right)
=
B \left( W,{W}^{\perp} \right)
=
\left\{ 0 \right\},
\end{equation*}
therefore ${W}^{\perp}$ is also an invariant subspace of ${\mathrm{Ad}}_{H}$. 
Let ${P}_{W}$ be one of the two projectors corresponding to the decomposition $\mathfrak{k} = W \oplus {W}^{\perp}$. 
Since $W$ and ${W}^{\perp}$ are both invariant subspaces of ${\mathrm{Ad}}_{H}$
\begin{equation*}
{P}_{W}
\in
{C}_{\mathfrak{gl \left( k \right)}} \left( {\mathrm{Ad}}_{H} \right)
=
{C}_{\mathfrak{gl \left( k \right)}} \left( {\mathrm{ad}}_{\mathfrak{g}} \right),
\end{equation*}
therefore $W$ and ${W}^{\perp}$ are also invariant subspaces of ${\mathrm{ad}}_{\mathfrak{g}}$.
In particular $\left[ \mathfrak{g},W \right] \subseteq W$.
\end{proof}
\begin{lemma}\label{finite}
Let $K \coloneqq SU \left( d \right)$ and $\mathfrak{k}=Lie(K)$. Let $\mathfrak{g}$ be a subalgebra of $\mathfrak{k}$ and let $H$ be a subgroup of $K$. 

\begin{equation*}
H \text{ is finite and } {C}_{\mathfrak{gl \left( k \right)}} \left( {\mathrm{Ad}}_{H} \right) = {C}_{\mathfrak{gl \left( k \right)}} \left( {\mathrm{ad}}_{\mathfrak{g}} \right)
\implies
H \cap {\mathcal{B}}_{K} \text{ and } \mathfrak{g} \text{ commute }.
\end{equation*}
\end{lemma}
\begin{proof}
By \cref{aisabelian} we know that $\mathfrak{a}(H)$ is abelian subalgebra of $\mathfrak{k}$ that is invariant under $\mathrm{Ad}_H$. Moreover, the equality  ${C}_{\mathfrak{gl \left( k \right)}} \left( {\mathrm{Ad}}_{H} \right) = {C}_{\mathfrak{gl \left( k \right)}} \left( {\mathrm{ad}}_{\mathfrak{g}} \right)$ combined with  \cref{invariantsubspace} implies that $\left[ \mathfrak{g},\mathfrak{a}(H) \right] \subseteq \mathfrak{a}(H)$.

We next consider the Killing form of $\mathfrak{k}$. Since
\begin{equation*}
K \left( \left[ \mathfrak{g},\mathfrak{a}(H) \right],\mathfrak{a}(H) \right)
=
K \left( \mathfrak{g},\left[ \mathfrak{a}(H),\mathfrak{a}(H) \right] \right)
=
\left\{ 0 \right\},
\end{equation*}
we conclude that $\left[ \mathfrak{g},\mathfrak{a}(H) \right] \subseteq {\mathfrak{a}(H)}^{\perp}$. Combining this with $\left[ \mathfrak{g},\mathfrak{a}(H) \right] \subseteq \mathfrak{a}(H)$ we obtain $\left[ \mathfrak{g},\mathfrak{a}(H) \right] \subseteq \mathfrak{a}(H) \cap {\mathfrak{a}}(H)^{\perp}$, and since $K$ is definite, $\left[ \mathfrak{g},\mathfrak{a} \right] = \left\{ 0 \right\}$. Using definition of $\mathfrak{a}(H)$ it is easy to see that this implies $\mathfrak{g}$ and $H \cap {\mathcal{B}}_{K}$ commute.
\end{proof}

\subsection{Universality Problem}
We are now ready to provide a proof of the universality theorem for sets of qudit gates  $S\subset SU(d)$.
\begin{theorem}\label{universalityproblemgroup}
Let $K \coloneqq SU \left( d \right)$, $\mathfrak{k}=Lie(K)$, $S$ be a subset of $K$ and $H \coloneqq \overline{\left\langle S \right\rangle}$. Then
\begin{equation*}
H = K
\iff
{C}_{\mathfrak{gl \left( k \right)}} \left( {\mathrm{Ad}}_{S} \right) = \left\{ \lambda \mathbb{1} \mid \lambda \in \mathbb{R} \right\} \text{ and } H \cap {\mathcal{B}}_{K} \not\subseteq {C}_{K} \left( K \right).
\end{equation*}
\end{theorem}
\begin{proof}
By \cref{centralizergroupgenerators} and \cref{centralizerclosure}
\begin{equation*}
{C}_{\mathfrak{gl \left( k \right)}} \left( {\mathrm{Ad}}_{S} \right)
=
{C}_{\mathfrak{gl \left( k \right)}} \left( {\mathrm{Ad}}_{\left\langle S \right\rangle} \right)
=
{C}_{\mathfrak{gl \left( k \right)}} \left( {\mathrm{Ad}}_{\overline{\left\langle S \right\rangle}} \right)
=
{C}_{\mathfrak{gl \left( k \right)}} \left( {\mathrm{Ad}}_{H} \right),
\end{equation*}
and by \cref{centralizercorrespondence}
\begin{equation*}
{C}_{\mathfrak{gl \left( k \right)}} \left( {\mathrm{ad}}_{\mathfrak{k}} \right)
=
{C}_{\mathfrak{gl \left( k \right)}} \left( {\mathrm{Ad}}_{K} \right).
\end{equation*}
\begin{description}
\item[$\Longrightarrow$]
\begin{equation*}
{C}_{\mathfrak{gl \left( k \right)}} \left( {\mathrm{Ad}}_{S} \right)
=
{C}_{\mathfrak{gl \left( k \right)}} \left( {\mathrm{Ad}}_{H} \right)
=
{C}_{\mathfrak{gl \left( k \right)}} \left( {\mathrm{Ad}}_{K} \right)
=
{C}_{\mathfrak{gl \left( k \right)}} \left( {\mathrm{ad}}_{\mathfrak{k}} \right)
=
\left\{ \lambda \mathbb{1} \mid \lambda \in \mathbb{R} \right\},
\end{equation*}
and
\begin{equation*}
H \cap {\mathcal{B}}_{K}
=
K \cap {\mathcal{B}}_{K}
=
{\mathcal{B}}_{K}
\not\subseteq
{C}_{K} \left( K \right).
\end{equation*}
\item[$\Longleftarrow$]
\begin{equation*}
{C}_{\mathfrak{gl \left( k \right)}} \left( {\mathrm{Ad}}_{H} \right)
=
{C}_{\mathfrak{gl \left( k \right)}} \left( {\mathrm{Ad}}_{S} \right)
=
\left\{ \lambda \mathbb{1} \mid \lambda \in \mathbb{R} \right\}
=
{C}_{\mathfrak{gl \left( k \right)}} \left( {\mathrm{ad}}_{\mathfrak{k}} \right).
\end{equation*}
Let $\mathfrak{h}$ be the Lie algebra of $H$. Since $\mathfrak{h}$ is an invariant subspace of ${\mathrm{Ad}}_{H}$, by \cref{invariantsubspace} $\mathfrak{h}$ is an ideal of $\mathfrak{k}$, and since $\mathfrak{k}$ is simple, either $\mathfrak{h} = \left\{ 0 \right\}$ or $\mathfrak{h} = \mathfrak{k}$. Suppose by contradiction $\mathfrak{h} = \left\{ 0 \right\}$. Then $H$ is finite, and by \cref{finite} (with $\mathfrak{g}=\mathfrak{k}$) we get $H \cap {\mathcal{B}}_{K} \text{ and } \mathfrak{k}$ commute. Hence  $H \cap {\mathcal{B}}_{K} \text{ and } K$ commute. Thus $H \cap {\mathcal{B}}_{K} \subseteq {C}_{K} \left( K \right)$, which is a contradiction. Therefore $\mathfrak{h} = \mathfrak{k}$, hence $H = K$.

\end{description}
\end{proof}
\begin{remark}
The proof uses the fact that $\mathfrak{k}$ is simple, therefore the theorem does not extend to $K = U \left( d \right)$.
\end{remark}
\begin{remark}
Since we do not have access to $H$, neither we have to $H \cap {\mathcal{B}}_{K}$, therefore \cref{universalityproblemgroup} is not directly applicable, unless there exists some integer $n$ such that
\begin{equation*}
H \cap {\mathcal{B}}_{K} \not\subseteq {C}_{K} \left( K \right)
\iff
{S}^{n} \cap {\mathcal{B}}_{K} \not\subseteq {C}_{K} \left( K \right).
\end{equation*}
As it has been shown in \cite{Adam}, it is possible to use the so-called \emph{spectral gap} of the gate set $S$ to upper bound this integer $n$.
\end{remark}

\begin{remark}
In the last part of the proof of \cref{universalityproblemgroup} we could also use the following argument: Then $H$ is finite, and by \cref{aisabelian} the space $\mathfrak{a}(H)$ is an abelian subalgebra of $\mathfrak{k}$ that is invariant under $\mathrm{Ad}_H$. This implies $\mathfrak{a}(H)=\{0\}$ or $\mathfrak{a}(H)=\mathfrak{k}$. The latter is impossible since $\mathfrak{k}$ is nonabelian. Therefore  $\mathfrak{a}(H)=\{0\}$ which by \cref{defofa} implies $H \cap {\mathcal{B}}_{K} \subseteq {C}_{K} \left( K \right)$, which is a contradiction. Therefore $\mathfrak{h} = \mathfrak{k}$, hence $H = K$.
\end{remark}
\subsection{Subgroup Universality Problem}
We start with a few facts related to the subalgebras - Lie subgroups correspondence.
\begin{lemma}\label{correspondence}
Let $K$ be a Lie group and $\mathfrak{k}=Lie(K)$. Let $\mathfrak{g}$ be a subalgebra of $\mathfrak{k}$. Let $G$ be the connected Lie subgroup of $K$ such that $Lie(G)=\mathfrak{g}$. Then
\begin{equation*}
G
=
\left\langle \exp \left( \mathfrak{g} \right) \right\rangle:=\langle\{\exp(X)|\,X\in \mathfrak{g}\}\rangle.
\end{equation*}
\end{lemma}
\begin{lemma}\label{centralizercorrespondence}
Let $K$ be a Lie group. 
Let $\mathfrak{k}$ be the Lie algebra of $K$. 
Let $\mathfrak{g}$ be a subalgebra of $\mathfrak{k}$. 
Let $G$ be the connected Lie subgroup of $K$ corresponding to $\mathfrak{g}$.
\begin{equation*}
{C}_{\mathfrak{gl \left( k \right)}} \left( {\mathrm{ad}}_{\mathfrak{g}} \right)
=
{C}_{\mathfrak{gl \left( k \right)}} \left( {\mathrm{Ad}}_{G} \right).
\end{equation*}
\end{lemma}
\begin{proof}
By \cref{centralizerexponential}, \cref{centralizergroupgenerators} and \cref{correspondence}
\begin{equation*}
{C}_{\mathfrak{gl \left( k \right)}} \left( {\mathrm{ad}}_{\mathfrak{g}} \right)
=
{C}_{\mathfrak{gl \left( k \right)}} \left( {e}^{{\mathrm{ad}}_{\mathfrak{g}}} \right)
=
{C}_{\mathfrak{gl \left( k \right)}} \left( {\mathrm{Ad}}_{exp \left( \mathfrak{g} \right)} \right)
=
{C}_{\mathfrak{gl \left( k \right)}} \left( {\mathrm{Ad}}_{\left\langle exp \left( \mathfrak{g} \right) \right\rangle} \right)
=
{C}_{\mathfrak{gl \left( k \right)}} \left( {\mathrm{Ad}}_{G} \right).
\end{equation*}
\end{proof}
\begin{lemma}\label{Mostow}
Let $K$ be a \emph{compact} Lie group and  $\mathfrak{k}=Lie(K)$. Let $\mathfrak{g}$ be a subalgebra of $\mathfrak{k}$. 
Let $G$ be the connected Lie subgroup of $K$ corresponding to $\mathfrak{g}$. Then
\begin{equation*}
\mathfrak{g} \text{ is semisimple }
\implies
G \text{ is closed }.
\end{equation*}
\end{lemma}
\begin{proof}
Corollary 1 in \cite{Mostow}.
\end{proof}

After this short technical introduction we are ready to deal with the subgroup universality problem. The setting is as follows.  We are given a set of hamiltonians $\mathcal{X}$ which generates a Lie algebra $\mathfrak{g}=\left\langle \mathcal{X} \right\rangle$. By \cref{correspondence} there is a unique connected subgroup $G$ satisfying $Lie(G)=\mathfrak{g}$. We will further assume $\mathfrak{g}$ is a simple Lie algebra. Thus by \cref{Mostow} the group $G$ is a closed subgroup of $K$. The group $G$ defines the set of gates that we are interested in. We assume, however, that due to some reasons we can only access Hamiltonians $\mathcal{Y}\subset \mathfrak{g}$ and that the available gates are 
\begin{equation}
    S=\langle\exp(\mathcal{Y})\rangle:=\langle\{\exp(Y)|\,Y\in\mathcal{Y}\}\rangle,
\end{equation}
where $\mathcal{Y}$ is a subset of $\mathfrak{g}$. Our goal is to provide conditions that need to be satisfied so that $\overline{\langle S\rangle}=G$. Naturally, in order to have equality of the groups we need to have ${C}_{\mathfrak{gl \left( k \right)}} \left( {\mathrm{Ad}}_{S} \right) = {C}_{\mathfrak{gl \left( k \right)}} \left( {\mathrm{ad}}_{\mathcal{X}} \right)$. The additional condition is once again determined by properties of $H \cap {\mathcal{B}}_{K}$.

\begin{theorem}\label{subgroupuniversalityproblem}
Let $K \coloneqq SU \left( d \right)$, $\mathfrak{k}=Lie(K)$, $\mathcal{X}$ a subset of $\mathfrak{k}$, $\mathfrak{g} \coloneqq \left\langle \mathcal{X} \right\rangle$, $G$ the connected Lie subgroup of $K$ corresponding to $\mathfrak{g}$, $\mathcal{Y}$ a subset of $\mathfrak{g}$, $S \coloneqq exp \left( \mathcal{Y} \right)$, $H \coloneqq \overline{\left\langle S \right\rangle}$. Finally assume $\mathfrak{g}$ is simple.
\begin{equation*}
H = G
\iff
{C}_{\mathfrak{gl \left( k \right)}} \left( {\mathrm{Ad}}_{S} \right) = {C}_{\mathfrak{gl \left( k \right)}} \left( {\mathrm{ad}}_{\mathcal{X}} \right), H \cap {\mathcal{B}}_{K} \text{ and } \mathcal{X} \text{ do not commute. }
\end{equation*}
\end{theorem}
\begin{proof}
By \cref{centralizergroupgenerators} and \cref{centralizerclosure}
\begin{equation*}
{C}_{\mathfrak{gl \left( k \right)}} \left( {\mathrm{Ad}}_{S} \right)
=
{C}_{\mathfrak{gl \left( k \right)}} \left( {\mathrm{Ad}}_{\left\langle S \right\rangle} \right)
=
{C}_{\mathfrak{gl \left( k \right)}} \left( {\mathrm{Ad}}_{\overline{\left\langle S \right\rangle}} \right)
=
{C}_{\mathfrak{gl \left( k \right)}} \left( {\mathrm{Ad}}_{H} \right).
\end{equation*}
Since any representation of a Lie algebra is a Lie homomorphism, by \cref{centralizeralgebragenerators}, \cref{homomorphismgenerators} and \cref{centralizercorrespondence}
\begin{equation*}
{C}_{\mathfrak{gl \left( k \right)}} \left( {\mathrm{ad}}_{\mathcal{X}} \right)
=
{C}_{\mathfrak{gl \left( k \right)}} \left( \left\langle {\mathrm{ad}}_{\mathcal{X}} \right\rangle \right)
=
{C}_{\mathfrak{gl \left( k \right)}} \left( {\mathrm{ad}}_{\left\langle \mathcal{X} \right\rangle} \right)
=
{C}_{\mathfrak{gl \left( k \right)}} \left( {\mathrm{ad}}_{\mathfrak{g}} \right)
=
{C}_{\mathfrak{gl \left( k \right)}} \left( {\mathrm{Ad}}_{G} \right).
\end{equation*}
\begin{description}
\item[$\Longrightarrow$]
\begin{equation*}
{C}_{\mathfrak{gl \left( k \right)}} \left( {\mathrm{Ad}}_{S} \right)
=
{C}_{\mathfrak{gl \left( k \right)}} \left( {\mathrm{Ad}}_{H} \right)
=
{C}_{\mathfrak{gl \left( k \right)}} \left( {\mathrm{Ad}}_{G} \right)
=
{C}_{\mathfrak{gl \left( k \right)}} \left( {\mathrm{ad}}_{\mathcal{X}} \right).
\end{equation*}
Suppose by contradiction that $H \cap {\mathcal{B}}_{K} = G \cap {\mathcal{B}}_{K}$ and $\mathcal{X}$ commute. 
Then $G \cap {\mathcal{B}}_{K}$ and $\mathfrak{g}$ commute, therefore ${\left[ G \cap {\mathcal{B}}_{K},G \right]}_{\bullet} = \left\{ \mathbb{1} \right\}$, namely $G \cap {\mathcal{B}}_{K} \subseteq {C}_{G} \left( G \right)$. 
Since $\mathfrak{g}$ is semisimple, ${C}_{G} \left( G \right)$ is finite, therefore $G \cap {\mathcal{B}}_{K}$ is finite. 
Again since $\mathfrak{g}$ is semisimple, $\mathfrak{g} \neq \left\{ 0 \right\}$, therefore $G$ is infinite, hence $G \cap {\mathcal{B}}_{K}$ is infinite, which is a contradiction.
\item[$\Longleftarrow$]
\begin{equation*}
{C}_{\mathfrak{gl \left( k \right)}} \left( {\mathrm{Ad}}_{H} \right)
=
{C}_{\mathfrak{gl \left( k \right)}} \left( {\mathrm{Ad}}_{S} \right)
=
{C}_{\mathfrak{gl \left( k \right)}} \left( {\mathrm{ad}}_{\mathcal{X}} \right)
=
{C}_{\mathfrak{gl \left( k \right)}} \left( {\mathrm{ad}}_{\mathfrak{g}} \right).
\end{equation*}
Let $\mathfrak{h}$ be the Lie algebra of $H$. Since $\mathfrak{h}$ is an invariant subspace of ${\mathrm{Ad}}_{H}$, by \cref{invariantsubspace} $\left[ \mathfrak{g},\mathfrak{h} \right] \subseteq \mathfrak{h}$. 
We note now that
\begin{equation*}
\left\langle S \right\rangle
=
\left\langle exp \left( \mathcal{Y} \right) \right\rangle
\subseteq
\left\langle exp \left( \mathfrak{g} \right) \right\rangle
=
G.
\end{equation*}
Since $\mathfrak{g}$ is semisimple, by \cref{Mostow} $G$ is closed, therefore by the monotonicity of the closure
\begin{equation*}
H
=
\overline{\left\langle S \right\rangle}
\subseteq
\overline{G}
=
G.
\end{equation*}
Therefore $\mathfrak{h} \subseteq \mathfrak{g}$. Thus $\mathfrak{h}$ is an ideal of $\mathfrak{g}$, and since $\mathfrak{g}$ is simple, either $\mathfrak{h} = \left\{ 0 \right\}$ or $\mathfrak{h} = \mathfrak{g}$. \\
Suppose by contradiction $\mathfrak{h} = \left\{ 0 \right\}$, then $H$ is finite, and by \cref{finite} $H \cap {\mathcal{B}}_{K}$ and $\mathfrak{g}$ commute. \\
Then $H \cap {\mathcal{B}}_{K}$ and $\mathcal{X}$ commute, which is a contradiction. Therefore $\mathfrak{h} = \mathfrak{g}$, hence $H = G$.
\end{description}
\end{proof}
\begin{remark}\label{no_u}
The theorem can be extended to $K = U \left( d \right)$, but there is little point in doing this.
The reason is that any \emph{perfect} subalgebra $\mathfrak{p}$ of a Lie algebra $\mathfrak{L}$ is contained in ${\mathfrak{L}}^{\prime}$ ($\mathfrak{p} = {\mathfrak{p}}^{\prime} \subseteq {\mathfrak{L}}^{\prime}$).
In our case $\mathfrak{g}$ is simple by assumption (hence perfect), therefore $\mathfrak{g} \subseteq {\mathfrak{k}}^{\prime} = \mathfrak{su} \left( d \right)$ even if $K = U \left( d \right)$.
\end{remark}
\begin{remark}
Since
\begin{equation*}
H \cap {\mathcal{B}}_{K} \text{ and } \mathcal{X} \text{ do not commute }
\iff
H \cap {\mathcal{B}}_{K} \not\subseteq {C}_{G} \left( G \right),
\end{equation*}
the reader might wonder why in \cref{subgroupuniversalityproblem} we use the LHS instead of the RHS, which is analogous to the condition in \cref{universalityproblemgroup}. \\
The reason is that we do not have access to ${C}_{G} \left( G \right)$, therefore \cref{subgroupuniversalityproblem} would not be directly applicable.
\end{remark}
\begin{remark}
We make a similar remark as for \cref{universalityproblemgroup}: since we do not have access to $H$, neither we have to $H \cap {\mathcal{B}}_{K}$, therefore \cref{subgroupuniversalityproblem} is not directly applicable, unless there exists some integer $n$ such that
\begin{equation*}
H \cap {\mathcal{B}}_{K} \text{ and } \mathcal{X} \text{ do not commute }
\iff
{S}^{n} \cap {\mathcal{B}}_{K} \text{ and } \mathcal{X} \text{ do not commute }.
\end{equation*}
It is possible that the spectral gap gives an upper bound in this case too.
\end{remark}
\begin{remark}
The only way to check if $\mathfrak{g}$ is simple is by first generating it, which we were trying to avoid.
\end{remark}
\begin{remark}
Does \cref{subgroupuniversalityproblem} hold if we relax the assumption that $\mathfrak{g}$ is simple? 
We have no clue when $\mathfrak{g}$ is semisimple, but when $\mathfrak{g}$ is \emph{not} semisimple we have found an involuted counterexample for $SU \left( 3 \right)$, which we do not report here.
\end{remark}
\subsection{Membership Problem}
We are given two finite sets of qu-$d$-it gates ${S}_{1}$ and $T$, and we ask if (and how) we can decide whether $T \subset \overline{\left\langle {S}_{1} \right\rangle}$.
Since $\overline{\left\langle {S}_{1} \right\rangle} \subseteq \overline{\left\langle {S}_{1} \cup T \right\rangle}$, we reformulate the problem and ask if we can decide whether $\overline{\left\langle {S}_{1} \right\rangle} = \overline{\left\langle {S}_{1} \cup T \right\rangle}$. \\
In the spirit of the theorems above, we start with the following necessary condition.
\begin{lemma}\label{necessarycondition}
Let $K \coloneqq SU \left( d \right)$, $\mathfrak{k}$ be the Lie algebra of $K$, 
Let ${S}_{1}$ and $T$ be subsets of $K$ and let ${S}_{2} \coloneqq {S}_{1} \cup T$. Finally let
Let ${H}_{1} \coloneqq \overline{\left\langle {S}_{1} \right\rangle}$ and ${H}_{2} \coloneqq \overline{\left\langle {S}_{2} \right\rangle}$.
\begin{equation*}
{H}_{1} = {H}_{2}
\implies
{C}_{\mathfrak{gl \left( k \right)}} \left( {\mathrm{Ad}}_{{S}_{1}} \right) = {C}_{\mathfrak{gl \left( k \right)}} \left( {\mathrm{Ad}}_{{S}_{2}} \right).
\end{equation*}
\end{lemma}
\begin{proof}
By \cref{centralizergroupgenerators} and \cref{centralizerclosure}
\begin{equation*}
{C}_{\mathfrak{gl \left( k \right)}} \left( {\mathrm{Ad}}_{{S}_{i}} \right)
=
{C}_{\mathfrak{gl \left( k \right)}} \left( {\mathrm{Ad}}_{\left\langle {S}_{i} \right\rangle} \right)
=
{C}_{\mathfrak{gl \left( k \right)}} \left( {\mathrm{Ad}}_{\overline{\left\langle {S}_{i} \right\rangle}} \right)
=
{C}_{\mathfrak{gl \left( k \right)}} \left( {\mathrm{Ad}}_{{H}_{i}} \right)
\quad
\left( i = 1,2 \right),
\end{equation*}
therefore
\begin{equation*}
{C}_{\mathfrak{gl \left( k \right)}} \left( {\mathrm{Ad}}_{{S}_{1}} \right)
=
{C}_{\mathfrak{gl \left( k \right)}} \left( {\mathrm{Ad}}_{{H}_{1}} \right)
=
{C}_{\mathfrak{gl \left( k \right)}} \left( {\mathrm{Ad}}_{{H}_{2}} \right)
=
{C}_{\mathfrak{gl \left( k \right)}} \left( {\mathrm{Ad}}_{{S}_{2}} \right).
\end{equation*}
\end{proof}
Once the necessary condition is satisfied we can use our techniques to solve the membership problem in a few particular cases. First we will assume that $\mathfrak{g}_1$ and $\mathfrak{g}_2$ are simple (in one case the assumption on $\mathfrak{g}_1$ is actually redundant). This is natural assumption from the point of view of \cref{subgroupuniversalityproblem}. We will also assume that $H_1 \cap {\mathcal{B}}_{K} \text{ and } \mathcal{X}_1$ do not commute which of course implies that $H_2 \cap {\mathcal{B}}_{K} \text{ and } \mathcal{X}_2$ do not commute. The problem splits into five cases. We are able to solve the membership problem for two inclusion diagrams below
\begin{align}\label{wecan}
\begin{matrix*}
{C}_{\mathfrak{gl \left( k \right)}} \left( {\mathrm{Ad}}_{{S}_{1}} \right) & = & {C}_{\mathfrak{gl \left( k \right)}} \left( {\mathrm{Ad}}_{{S}_{2}} \right)\qquad\\
\shortparallel &  &  \shortparallel\qquad\\
{C}_{\mathfrak{gl \left( k \right)}} \left( {\mathrm{ad}}_{{\mathcal{X}}_{1}} \right) &=&{C}_{\mathfrak{gl \left( k \right)}} \left( {\mathrm{ad}}_{{\mathcal{X}}_{2}} \right)\qquad
\end{matrix*}
&
\begin{matrix*}
{C}_{\mathfrak{gl \left( k \right)}} \left( {\mathrm{Ad}}_{{S}_{1}} \right) & = & {C}_{\mathfrak{gl \left( k \right)}} \left( {\mathrm{Ad}}_{{S}_{2}} \right)\qquad\\
\cup &  &  \shortparallel\qquad\\
{C}_{\mathfrak{gl \left( k \right)}} \left( {\mathrm{ad}}_{{\mathcal{X}}_{1}} \right) & \supset & {C}_{\mathfrak{gl \left( k \right)}} \left( {\mathrm{ad}}_{{\mathcal{X}}_{2}} \right)\qquad
\end{matrix*}.
\end{align}
For the next three diagrams our methods are insufficient to solve the membership problem.
\begin{scriptsize}
\begin{align}\label{wecant}
\begin{matrix}
{C}_{\mathfrak{gl \left( k \right)}} \left( {\mathrm{Ad}}_{{S}_{1}} \right) & = & {C}_{\mathfrak{gl \left( k \right)}} \left( {\mathrm{Ad}}_{{S}_{2}} \right)\qquad\\
\cup &  &  \cup\qquad\\
{C}_{\mathfrak{gl \left( k \right)}} \left( {\mathrm{ad}}_{{\mathcal{X}}_{1}} \right) & = & {C}_{\mathfrak{gl \left( k \right)}} \left( {\mathrm{ad}}_{{\mathcal{X}}_{2}} \right)\qquad
\end{matrix}
&
\begin{matrix}
{C}_{\mathfrak{gl \left( k \right)}} \left( {\mathrm{Ad}}_{{S}_{1}} \right) & = & {C}_{\mathfrak{gl \left( k \right)}} \left( {\mathrm{Ad}}_{{S}_{2}} \right)\qquad\\
\shortparallel &  &  \cup\qquad\\
{C}_{\mathfrak{gl \left( k \right)}} \left( {\mathrm{ad}}_{{\mathcal{X}}_{1}} \right) & \supset & {C}_{\mathfrak{gl \left( k \right)}} \left( {\mathrm{ad}}_{{\mathcal{X}}_{2}} \right)\qquad
\end{matrix}
&
\begin{matrix}
{C}_{\mathfrak{gl \left( k \right)}} \left( {\mathrm{Ad}}_{{S}_{1}} \right) & = & {C}_{\mathfrak{gl \left( k \right)}} \left( {\mathrm{Ad}}_{{S}_{2}} \right)\qquad\\
\cup &  &  \cup\qquad\\
{C}_{\mathfrak{gl \left( k \right)}} \left( {\mathrm{ad}}_{{\mathcal{X}}_{1}} \right) & \supset & {C}_{\mathfrak{gl \left( k \right)}} \left( {\mathrm{ad}}_{{\mathcal{X}}_{2}} \right)\qquad
\end{matrix}.
\end{align}
\end{scriptsize}
The next two theorems provide the solution for the cases given in (\ref{wecan}).
\begin{theorem}\label{membershipproblemgroup}
Define the same objects as in the previous lemma and let ${\mathcal{X}}_{1}$ and $\mathcal{Y}$ be subsets of $\mathfrak{k}$ such that ${S}_{1} = exp \left( {\mathcal{X}}_{1} \right)$ and $T = exp \left( \mathcal{Y} \right)$, and let ${\mathcal{X}}_{2} \coloneqq {\mathcal{X}}_{1} \cup \mathcal{Y}$.Let $\mathfrak{{g}_{1}} \coloneqq \left\langle {\mathcal{X}}_{1} \right\rangle$ and $\mathfrak{{g}_{2}} \coloneqq \left\langle {\mathcal{X}}_{2} \right\rangle$. Let ${G}_{1}$ and ${G}_{2}$ be the connected Lie subgroups of $K$ corresponding to $\mathfrak{{g}_{1}}$ and $\mathfrak{{g}_{2}}$, respectively. Finally assume $\mathfrak{{g}_{2}}$ is simple. Then
\begin{equation*}
\begin{matrix*}
{C}_{\mathfrak{gl \left( k \right)}} \left( {\mathrm{Ad}}_{{S}_{1}} \right) & = & {C}_{\mathfrak{gl \left( k \right)}} \left( {\mathrm{Ad}}_{{S}_{2}} \right)\qquad\\
\shortparallel &  &  \shortparallel\qquad\\
{C}_{\mathfrak{gl \left( k \right)}} \left( {\mathrm{ad}}_{{\mathcal{X}}_{1}} \right) &=&{C}_{\mathfrak{gl \left( k \right)}} \left( {\mathrm{ad}}_{{\mathcal{X}}_{2}} \right)\qquad
\end{matrix*}
, 
{H}_{1} \cap {\mathcal{B}}_{K} \text{ and } {\mathcal{X}}_{1} \text{ do not commute }
\implies
{H}_{1} = {H}_{2}.
\end{equation*}
\end{theorem}
\begin{proof}
Since $\mathfrak{{g}_{2}}$ is semisimple and
\begin{equation*}
{C}_{\mathfrak{gl \left( k \right)}} \left( {\mathrm{ad}}_{{\mathcal{X}}_{1}} \right)
=
{C}_{\mathfrak{gl \left( k \right)}} \left( {\mathrm{ad}}_{{\mathcal{X}}_{2}} \right),
\end{equation*}
by \cref{membershipproblemalgebra} and \cref{membershipremark} $\mathfrak{{g}_{1}} = \mathfrak{{g}_{2}}$. 
Since ${H}_{1} \cap {\mathcal{B}}_{K}$ and ${\mathcal{X}}_{1}$ do not commute, neither ${H}_{2} \cap {\mathcal{B}}_{K}$ and ${\mathcal{X}}_{2}$ do, and since $\mathfrak{{g}_{1}} = \mathfrak{{g}_{2}}$ is simple and
\begin{equation*}
\begin{split}
{C}_{\mathfrak{gl \left( k \right)}} \left( {\mathrm{Ad}}_{{S}_{1}} \right)
& =
{C}_{\mathfrak{gl \left( k \right)}} \left( {\mathrm{ad}}_{{\mathcal{X}}_{1}} \right)
\\
{C}_{\mathfrak{gl \left( k \right)}} \left( {\mathrm{Ad}}_{{S}_{2}} \right)
& =
{C}_{\mathfrak{gl \left( k \right)}} \left( {\mathrm{ad}}_{{\mathcal{X}}_{2}} \right),
\end{split}
\end{equation*}
by \cref{subgroupuniversalityproblem} ${H}_{1} = {G}_{1} = {G}_{2} = {H}_{2}$.
\end{proof}
\begin{theorem}\label{membershipproblemgroup2}
Define the same objects as in the previous lemma and let ${\mathcal{X}}_{1}$ and $\mathcal{Y}$ be subsets of $\mathfrak{k}$ such that ${S}_{1} = exp \left( {\mathcal{X}}_{1} \right)$ and $T = exp \left( \mathcal{Y} \right)$, and let ${\mathcal{X}}_{2} \coloneqq {\mathcal{X}}_{1} \cup \mathcal{Y}$.Let $\mathfrak{{g}_{1}} \coloneqq \left\langle {\mathcal{X}}_{1} \right\rangle$ and $\mathfrak{{g}_{2}} \coloneqq \left\langle {\mathcal{X}}_{2} \right\rangle$. Let ${G}_{1}$ and ${G}_{2}$ be the connected Lie subgroups of $K$ corresponding to $\mathfrak{{g}_{1}}$ and $\mathfrak{{g}_{2}}$, respectively. Finally assume $\mathfrak{{g}_{1}}$ and  $\mathfrak{{g}_{2}}$ are simple. Then
\begin{equation*}
\begin{matrix*}
{C}_{\mathfrak{gl \left( k \right)}} \left( {\mathrm{Ad}}_{{S}_{1}} \right) & = & {C}_{\mathfrak{gl \left( k \right)}} \left( {\mathrm{Ad}}_{{S}_{2}} \right)\qquad\\
\cup &  &  \shortparallel\qquad\\
{C}_{\mathfrak{gl \left( k \right)}} \left( {\mathrm{ad}}_{{\mathcal{X}}_{1}} \right) & \supset & {C}_{\mathfrak{gl \left( k \right)}} \left( {\mathrm{ad}}_{{\mathcal{X}}_{2}} \right)\qquad
\end{matrix*}
, 
{H}_{1} \cap {\mathcal{B}}_{K} \text{ and } {\mathcal{X}}_{1} \text{ do not commute }
\implies
{H}_{1} \neq {H}_{2}.
\end{equation*}
\end{theorem}
\begin{proof}
Since
\begin{equation*}
{C}_{\mathfrak{gl \left( k \right)}} \left( {\mathrm{ad}}_{{\mathcal{X}}_{1}} \right)
\neq
{C}_{\mathfrak{gl \left( k \right)}} \left( {\mathrm{ad}}_{{\mathcal{X}}_{2}} \right),
\end{equation*}
by \cref{membershipproblemalgebra} $\mathfrak{{g}_{1}} \subset \mathfrak{{g}_{2}}$, therefore ${G}_{1} \subset {G}_{2}$. 
Since ${H}_{1} \cap {\mathcal{B}}_{K}$ and ${\mathcal{X}}_{1}$ do not commute, neither ${H}_{2} \cap {\mathcal{B}}_{K}$ and ${\mathcal{X}}_{2}$ do, and since $\mathfrak{{g}_{2}}$ is simple and
\begin{equation*}
{C}_{\mathfrak{gl \left( k \right)}} \left( {\mathrm{Ad}}_{{S}_{2}} \right)
=
{C}_{\mathfrak{gl \left( k \right)}} \left( {\mathrm{ad}}_{{\mathcal{X}}_{2}} \right),
\end{equation*}
by \cref{subgroupuniversalityproblem} ${H}_{2} = {G}_{2}$.
We note now that
\begin{equation*}
\left\langle {S}_{1} \right\rangle
=
\left\langle exp \left( {\mathcal{X}}_{1} \right) \right\rangle
\subseteq
\left\langle exp \left( \mathfrak{{g}_{1}} \right) \right\rangle
=
{G}_{1}.
\end{equation*}
Since $\mathfrak{{g}_{1}}$ is semisimple, by \cref{Mostow} ${G}_{1}$ is closed, therefore by the monotonicity of the closure
\begin{equation*}
{H}_{1}
=
\overline{\left\langle {S}_{1} \right\rangle}
\subseteq
\overline{{G}_{1}}
=
{G}_{1}.
\end{equation*}
Putting all together, ${H}_{1} \subseteq {G}_{1} \subset {G}_{2} = {H}_{2}$.
\end{proof}
\begin{remark}
If we keep in mind that a necessary condition for ${H}_{1} = {H}_{2}$ is that ${H}_{2} \subseteq {G}_{1}$, we see why we cannot solve the cases corresponding to the diagrams in \eqref{wecant}: since ${H}_{2} \subset {G}_{2}$, if ${G}_{1} = {G}_{2}$ the necessary condition is trivially satisfied but then we stop, if otherwise ${G}_{1} \subset {G}_{2}$ we cannot test it with our methods.
\end{remark}
\begin{remark}
The theorems can be extended to $K = U \left( d \right)$, but again \cref{no_u} shows that we gain nothing by doing this.
\end{remark}
\section{Conclusions}
We have dealt with the Universality and Membership Problems for Lie algebras and Lie groups, together with the Subgroup Universality Problem.
The theorems for the Lie algebraic problems directly provide algorithms for solving them, but unfortunately the same cannot be said for the other problems (as already clarified in a handful of remarks).
We are indeed either making additional assumptions, like the simplicity of $\mathfrak{g}$ in the Subgroup Universality Problem, or relying on conditions which cannot be checked in general, like in the Universality Problem for Lie groups.
Moreover, even in such a simplified setting, we still have to distinguish ``good'' and ``bad'' cases in the Membership Problem for Lie groups.
Trying to overcome these limitations is therefore the most important open question which is left to future explorations.
\section{Acknowledgements}
This work was supported by National Science Centre, Poland under the grant SONATA BIS:
2015/18/E/ST1/00200
\bibliographystyle{unsrt}
\bibliography{Paper}
\addcontentsline{toc}{section}{References}
\section{Appendix}
\subsection{Prerequisites for Lie algebras}
We recall some basic definitions for a Lie algebra $\mathfrak{g}$:
$Rad \left( \mathfrak{g} \right)$, the \emph{radical} of $\mathfrak{g}$, is the maximal solvable ideal of $\mathfrak{g}$;
${\mathfrak{g}}^{\prime}$, the \emph{derived algebra} of $\mathfrak{g}$, is a shorthand for $\left[ \mathfrak{g},\mathfrak{g} \right]$;
$\mathrm{ad} : \mathfrak{g} \rightarrow \mathfrak{gl \left( g \right)}$, the \emph{adjoint representation} of $\mathfrak{g}$, is defined by ${\mathrm{ad}}_{x} \coloneqq \mathrm{ad} \left( x \right) = \left[ x,\cdot \right]$.
We start with the following inclusions.
\begin{lemma}\label{radicals}
Let $\mathfrak{g}$ be a Lie algebra.
\begin{equation*}
{C}_{\mathfrak{g}} \left( \mathfrak{g} \right)
\subseteq
{\mathfrak{g}}^{\perp}
\subseteq
Rad \left( \mathfrak{g} \right).
\end{equation*}
\end{lemma}
We are interested in the case in which the chain of inclusions ``collapses''.
\begin{definition}\label{reductivealgebra}
Let $\mathfrak{g}$ be a Lie algebra. 
$\mathfrak{g}$ is \emph{reductive} when $Rad \left( \mathfrak{g} \right) = {C}_{\mathfrak{g}} \left( \mathfrak{g} \right)$.
\end{definition}
Among the nice properties of reductive Lie algebras, we are interested in the following decomposition.
\begin{theorem}\label{reductivedecomposition}
Let $\mathfrak{g}$ be a \emph{reductive} Lie algebra.
\begin{equation*}
\mathfrak{g}
=
{C}_{\mathfrak{g}} \left( \mathfrak{g} \right) \oplus {\mathfrak{g}}^{\prime}
\text{ and }
{\mathfrak{g}}^{\prime} \text{ is semisimple }.
\end{equation*}
\end{theorem}
If we consider \emph{real} Lie algebras, we can define an even more interesting kind of Lie algebras.
\begin{definition}\label{compactalgebra}

Let $\mathfrak{g}$ be a real Lie algebra. 
$\mathfrak{g}$ is \emph{compact} when any of the following equivalent conditions holds:
\begin{itemize}
\item $\mathfrak{g}$ is reductive and the Killing form of $\mathfrak{g}$ is negative semidefinite
\item there exists a compact Lie group with Lie algebra $\mathfrak{g}$
\item there exists a definite symmetric bilinear form $B$ such that $\forall \ x,y,z \in \mathfrak{g}$ $B \left( \left[ x,y \right],z \right) = B \left( x, \left[ y,z \right] \right)$
\end{itemize}
\end{definition}
\begin{proof}
Page 284 in \cite{Bourbaki}.
\end{proof}
Compact Lie algebras behave well under taking subalgebras.
\begin{theorem}\label{compactsubalgebra}
Subalgebras of a \emph{compact} Lie algebra are also compact.
\end{theorem}
\begin{proof}
Page 284 in \cite{Bourbaki}.
\end{proof}
\begin{remark}\label{reductivesubalgebra}
By the combination of \cref{compactsubalgebra}, the first condition in \cref{compactalgebra} and \cref{reductivedecomposition}, any subalgebra $\mathfrak{g}$ of a \emph{compact} Lie algebra $\mathfrak{k}$ decomposes as $\mathfrak{g} = {C}_{\mathfrak{g}} \left( \mathfrak{g} \right) \oplus {\mathfrak{g}}^{\prime}$. 
Since we will use this fact over and over again without mentioning, it is useful to clarify it once for all.
\end{remark}
We also have a useful characterization at the level of Lie groups.
\begin{theorem}\label{compactgroup}
Let $G$ be a Lie group whose group of connected components is finite. 
Let $\mathfrak{g}$ be the Lie algebra of $G$. 
The following conditions are equivalent:
\begin{itemize}
\item $\mathfrak{g}$ is compact
\item there exists a definite symmetric bilinear form $B$ such that $\forall \ g \in G$ $\forall \ x,y \in \mathfrak{g}$ $B \left( {\mathrm{Ad}}_{g} x,{\mathrm{Ad}}_{g} y \right) = B \left( x,y \right)$
\end{itemize}
\end{theorem}
\begin{proof}
Page 284 in \cite{Bourbaki}.
\end{proof}
We state now without proof the following not-well-known results of linear algebra.
\begin{lemma}\label{doubleperp}
Let $V$ be a real vector space. 
Let $\left( \cdot,\cdot \right)$ be a symmetric bilinear form.
Let $\bar{V}$ be a subspace of $V$ such that $V = {V}^{\perp} \oplus \bar{V}$. 
Let $W$ be a subspace of $V$.
\begin{equation*}
{\left( \bar{V} \cap {W}^{\perp} \right)}^{\perp}
=
{V}^{\perp} + W.
\end{equation*}
\end{lemma}
\begin{lemma}\label{semidefinite}
Let $V$ be a real vector space. 
Let $\left( \cdot,\cdot \right)$ be a \emph{semidefinite} symmetric bilinear form. 
Let $W$ be a subspace of $V$.
\begin{equation*}
V = W \oplus {W}^{\perp}
\iff
{V}^{\perp} \cap W = \left\{ 0 \right\}.
\end{equation*}
\end{lemma}
We then apply \cref{doubleperp} and \cref{semidefinite} for Lie algebras.
\begin{lemma}\label{centralizerperp}

Let $\mathfrak{k}$ be a \emph{reductive} Lie algebra. 
Let $\mathcal{X}$ be a subset of $\mathfrak{k}$.
\begin{equation*}
{\left( {\mathfrak{k}}^{\prime} \cap {{C}_{\mathfrak{k}} \left( \mathcal{X} \right)}^{\perp} \right)}^{\perp}
=
{C}_{\mathfrak{k}} \left( \mathcal{X} \right).
\end{equation*}
\end{lemma}
\begin{proof}
Since $\mathfrak{k}$ is reductive, $\mathfrak{k} = {\mathfrak{k}}^{\perp} \oplus {\mathfrak{k}}^{\prime}$ and ${\mathfrak{k}}^{\perp} = {C}_{\mathfrak{k}} \left( \mathfrak{k} \right)$, therefore by \cref{doubleperp}
\begin{equation*}
{\left( {\mathfrak{k}}^{\prime} \cap {{C}_{\mathfrak{k}} \left( \mathcal{X} \right)}^{\perp} \right)}^{\perp}
=
{\mathfrak{k}}^{\perp} + {C}_{\mathfrak{k}} \left( \mathcal{X} \right)
=
{C}_{\mathfrak{k}} \left( \mathfrak{k} \right) + {C}_{\mathfrak{k}} \left( \mathcal{X} \right)
=
{C}_{\mathfrak{k}} \left( \mathcal{X} \right).
\end{equation*}
\end{proof}
\begin{theorem}\label{compactdecomposition}

Let $\mathfrak{k}$ be a \emph{compact} Lie algebra. 
Let $\mathcal{X}$ be a subset of $\mathfrak{k}$.
\begin{equation*}
\mathfrak{k}
=
{C}_{\mathfrak{k}} \left( \mathcal{X} \right) \oplus {\mathfrak{k}}^{\prime} \cap {{C}_{\mathfrak{k}} \left( \mathcal{X} \right)}^{\perp}.
\end{equation*}
\end{theorem}
\begin{proof}

Since $\mathfrak{k}$ is reductive
\begin{equation*}
{\mathfrak{k}}^{\perp} \cap \left( {\mathfrak{k}}^{\prime} \cap {{C}_{\mathfrak{k}} \left( \mathcal{X} \right)}^{\perp} \right)
=
\left( {\mathfrak{k}}^{\perp} \cap {\mathfrak{k}}^{\prime} \right) \cap {{C}_{\mathfrak{k}} \left( \mathcal{X} \right)}^{\perp}
=
\left\{ 0 \right\}.
\end{equation*}
Since the Killing form of $\mathfrak{k}$ is semidefinite, by \cref{semidefinite}
\begin{equation*}
\mathfrak{k}
=
{\left( {\mathfrak{k}}^{\prime} \cap {{C}_{\mathfrak{k}} \left( \mathcal{X} \right)}^{\perp} \right)}^{\perp} \oplus {\mathfrak{k}}^{\prime} \cap {{C}_{\mathfrak{k}} \left( \mathcal{X} \right)}^{\perp}.
\end{equation*}
Again since $\mathfrak{k}$ is reductive, by \cref{centralizerperp}
\begin{equation*}
\mathfrak{k}
=
{C}_{\mathfrak{k}} \left( \mathcal{X} \right) \oplus {\mathfrak{k}}^{\prime} \cap {{C}_{\mathfrak{k}} \left( \mathcal{X} \right)}^{\perp}.
\end{equation*}
\end{proof}
\subsection{Prerequisites for Lie groups}
We recall the so-called adjoint representation $\mathrm{Ad}:K\rightarrow GL(\mathfrak{k})$ which for matrix Lie groups is defined by the conjugation $\mathrm{Ad}_g(X):=gXg^{-1}$.
We need the following lemmas.
\begin{lemma}\label{parameter}

Let $V$ be a real vector space. 
Let $A$ and $B$ be endomorphisms of $V$.
\begin{equation*}
\left[ A,B \right] = 0
\iff
\forall \ t \in \mathbb{R} \left[ A,{e}^{t B} \right] = 0.
\end{equation*}
\end{lemma}
\begin{lemma}\label{centralizerexponential}
Let $\mathfrak{k}$ be a Lie algebra. 
Let $W$ be a subspace of $\mathfrak{k}$.
\begin{equation*}
{C}_{\mathfrak{gl \left( k \right)}} \left( {\mathrm{ad}}_{W} \right)
=
{C}_{\mathfrak{gl \left( k \right)}} \left( {e}^{{\mathrm{ad}}_{W}} \right).
\end{equation*}
\end{lemma}
\begin{proof}
The inclusion ${C}_{\mathfrak{gl \left( k \right)}} \left( {\mathrm{ad}}_{W} \right) \subseteq {C}_{\mathfrak{gl \left( k \right)}} \left( {e}^{{\mathrm{ad}}_{W}} \right)$ is trivial.
For the inclusion ${C}_{\mathfrak{gl \left( k \right)}} \left( {e}^{{\mathrm{ad}}_{W}} \right) \subseteq {C}_{\mathfrak{gl \left( k \right)}} \left( {\mathrm{ad}}_{W} \right)$,
$\forall \ A \in {C}_{\mathfrak{gl \left( k \right)}} \left( {e}^{{\mathrm{ad}}_{W}} \right)$
\begin{equation*}
\begin{split}
& \forall \ x \in W \left[ A,{e}^{{\mathrm{ad}}_{x}} \right] = 0, \\
& \forall \ y \in W \forall \ t \in \mathbb{R} \left[ A,{e}^{{\mathrm{ad}}_{t y}} \right] = 0, \\
& \forall \ y \in W \forall \ t \in \mathbb{R} \left[ A,{e}^{t {\mathrm{ad}}_{y}} \right] = 0.
\end{split}
\end{equation*}
By \cref{parameter}
\begin{equation*}
\forall \ y \in W \left[ A,{\mathrm{ad}}_{y} \right] = 0,
\end{equation*}
namely
\begin{equation*}
A \in {C}_{\mathfrak{gl \left( k \right)}} \left( {\mathrm{ad}}_{W} \right).
\end{equation*}
\end{proof}
We introduce now the reader to the main machinery of this work.
\begin{lemma}\label{inequality1}
Let $A$ and $B$ be unitary matrices.
\begin{equation*}
\left\| {\left[ A,B \right]}_{\bullet} - \mathbb{1} \right\|
\leq
\sqrt{2} \left\| A - \mathbb{1} \right\| \left\| B - \mathbb{1} \right\|,
\end{equation*}
%
\begin{equation*}
{\left[ A,{\left[ A,B \right]}_{\bullet} \right]}_{\bullet} = \mathbb{1} \quad \text{ and } \quad \left\| B - \mathbb{1} \right\| < 2
\implies
{\left[ A,B \right]}_{\bullet} = \mathbb{1}.
\end{equation*}
\end{lemma}
\begin{proof}
Theorems 36.15 and 36.16 in \cite{Curtis}.
\end{proof}
We reformulate now \cref{inequality1}.
\begin{theorem}\label{inequality2}
Let $K$ be a subgroup of $U \left( d \right)$. 
Let $g,h \in K$.
\begin{equation*}
\left\| {\left[ g,h \right]}_{\bullet} - \mathbb{1} \right\|
\leq
\sqrt{2} \min_{{c}_{1} \in {C}_{K} \left( K \right)} \left\| g - {c}_{1} \right\| \min_{{c}_{2} \in {C}_{K} \left( K \right)} \left\| h - {c}_{2} \right\|,
\end{equation*}
%
\begin{equation*}
{\left[ h,{\left[ h,g \right]}_{\bullet} \right]}_{\bullet} = \mathbb{1} \quad \text{ and } \quad g \in {\mathcal{B}}_{K}
\implies
{\left[ h,g \right]}_{\bullet} = \mathbb{1}.
\end{equation*}
\end{theorem}
\begin{proof}
By \cref{inequality1}, $\forall \ {c}_{1},{c}_{2} \in {C}_{K} \left( K \right)$
\begin{equation*}
\left\| {\left[ g,h \right]}_{\bullet} - \mathbb{1} \right\|
=
\left\| {\left[ {{c}_{1}}^{- 1} g,{{c}_{2}}^{- 1} h \right]}_{\bullet} - \mathbb{1} \right\|
\leq
\sqrt{2} \left\| {{c}_{1}}^{- 1} g - \mathbb{1} \right\| \left\| {{c}_{2}}^{- 1} h - \mathbb{1} \right\|
=
\sqrt{2} \left\| g - {c}_{1} \right\| \left\| h - {c}_{2} \right\|,
\end{equation*}
therefore
\begin{equation*}
\left\| {\left[ g,h \right]}_{\bullet} - \mathbb{1} \right\|
\leq
\sqrt{2} \min_{{c}_{1} \in {C}_{K} \left( K \right)} \left\| g - {c}_{1} \right\| \min_{{c}_{2} \in {C}_{K} \left( K \right)} \left\| h - {c}_{2} \right\|.
\end{equation*}
%
We also know that
\begin{equation*}
\exists \ c \in {C}_{K} \left( K \right) \mid g \in {\mathcal{B}}_{K} \left( c \right).
\end{equation*}
Now
\begin{equation*}
{\left[ h,{\left[ h,{c}^{- 1} g \right]}_{\bullet} \right]}_{\bullet}
=
{\left[ h,{\left[ h,g \right]}_{\bullet} \right]}_{\bullet}
=
\mathbb{1},
\end{equation*}
and
\begin{equation*}
\left\| {c}^{- 1} g - \mathbb{1} \right\|
=
\left\| g - c \right\|
<
\frac{1}{\sqrt{2}}
<
2.
\end{equation*}
By \cref{inequality1}
\begin{equation*}
{\left[ h,g \right]}_{\bullet}
=
{\left[ h,{c}^{- 1} g \right]}_{\bullet}
=
\mathbb{1}.
\end{equation*}
\end{proof}
\subsection{Logarithm of a special unitary matrix}\label{extra}
It is well known that the exponential of a traceless skew - hermitian matrix is a special unitary matrix.
In this section we ask under which conditions the logarithm of a special unitary matrix is a traceless skew - hermitian matrix, we recall how to define the logarithm of a matrix and we provide an answer to the question in the final theorem.
\begin{definition}
Let $A$ be an invertible matrix. 
The \emph{principal logarithm} of $A$ is the unique logarithm of $A$ with eigenvalues in $\pi< Im \left( z \right) \leq \pi$.
\end{definition}
We start with an upper bound on the trace of the logarithm.
\begin{lemma}
Let $U \in U \left( d \right)$ and $X \in \mathfrak{u} \left( d \right)$ be the principal logarithm of $U$.
\begin{equation*}
\text{ The eigenvalues of } U \text{ are in } Re \left( z \right) \geq 0
\implies
\left\lvert \frac{\mathrm{tr} \left( X \right)}{2 \pi i} \right\rvert \leq \frac{d}{2 \pi} \mathrm{acos} \left( 1 - \frac{{\left\| U - \mathbb{1} \right\|}^{2}}{2 d} \right) \leq \frac{d}{4}.
\end{equation*}
\end{lemma}
\begin{proof}
Since the eigenvalues of $U$ are of the form ${e}^{i {\theta}_{i}} \left( i = 1,\dots,d \right)$
\begin{equation*}
\begin{split}
& 1 - \frac{{\left\| U - \mathbb{1} \right\|}^{2}}{2 d}
=
\frac{\sum_{i = 1}^{d} \cos{{\theta}_{i}}}{d},
\\
& \frac{\mathrm{tr} \left( X \right)}{d i}
=
\frac{i \sum_{i = 1}^{d} {\theta}_{i}}{d i}
=
\frac{\sum_{i = 1}^{d} {\theta}_{i}}{d}.
\end{split}
\end{equation*}
Since $\left\lvert {\theta}_{i} \right\rvert \leq \frac{\pi}{2} \left( i = 1,\dots,d \right)$
\begin{itemize}
\item
$0 \leq 1 - \frac{{\left\| U - \mathbb{1} \right\|}^{2}}{2 d} \leq 1$,
\item
$\left\lvert \frac{\mathrm{tr} \left( X \right)}{d i} \right\rvert \leq \frac{\pi}{2}$,
\item
since the function $\cos{\theta}$ is concave for $\left\lvert \theta \right\rvert \leq \frac{\pi}{2}$, by Jensen's Inequality
\begin{equation*}
\cos{\frac{\mathrm{tr} \left( X \right)}{d i}}
=
\cos{\frac{\sum_{i = 1}^{d} {\theta}_{i}}{d}}
\geq
\frac{\sum_{i = 1}^{d} \cos{{\theta}_{i}}}{d}
=
1 - \frac{{\left\| U - \mathbb{1} \right\|}^{2}}{2 d}.
\end{equation*}
\end{itemize}
We can sum it up in the chain of inequalities
\begin{equation*}
\left\lvert \frac{\mathrm{tr} \left( X \right)}{d i} \right\rvert
\leq
\mathrm{acos} \left( 1 - \frac{{\left\| U - \mathbb{1} \right\|}^{2}}{2 d} \right)
\leq
\frac{\pi}{2},
\end{equation*}
or equivalently
\begin{equation*}
\left\lvert \frac{\mathrm{tr} \left( X \right)}{2 \pi i} \right\rvert
\leq
\frac{d}{2 \pi} \mathrm{acos} \left( 1 - \frac{{\left\| U - \mathbb{1} \right\|}^{2}}{2 d} \right)
\leq
\frac{d}{4}.
\end{equation*}
\end{proof}
We can now replace in the previous lemma the principal logarithm with the logarithm defined via matrix power series, at the cost of introducing a sufficient condition for the convergence of the series.
\begin{lemma}
Let $r \in \left.\left( 0,1 \right.\right]$ and $U \in U \left( d \right)$ such that $\left\| U - \mathbb{1} \right\| < r$.
\begin{equation*}
\left\lvert \frac{\mathrm{tr} \left( \log{\left( U \right)} \right)}{2 \pi i} \right\rvert
\leq
\frac{d}{2 \pi} \mathrm{acos} \left( 1 - \frac{{\left\| U - \mathbb{1} \right\|}^{2}}{2 d} \right)
<
\frac{d}{2 \pi} \mathrm{acos} \left( 1 - \frac{{r}^{2}}{2 d} \right).
\end{equation*}
%
%
\end{lemma}
We then introduce some additional conditions which simplify the upper bound.
\begin{lemma}
Let $r \in \left.\left( 0,1 \right.\right]$ and $U \in U \left( d \right)$ such that $\left\| U - \mathbb{1} \right\| < r$.
\begin{equation*}
d \geq \left\lceil \frac{40}{{r}^{2}} \right\rceil \quad \text{ and } \quad \left\| U - \mathbb{1} \right\| < 2 \sqrt{d} \sin{\frac{\pi}{d}} 
\implies
\left\lvert \frac{\mathrm{tr} \left( \log{\left( U \right)} \right)}{2 \pi i} \right\rvert < 1.
\end{equation*}
\end{lemma}
\begin{proof}
Since ${\pi}^{2} < 10$
\begin{equation*}
{\left( \frac{2 \pi}{d} \right)}^{2}
=
\frac{4 {\pi}^{2}}{{d}^{2}}
<
\frac{40}{{d}^{2}}
=
\frac{{r}^{2}}{{d}^{2}} \frac{40}{{r}^{2}}
\leq
\frac{{r}^{2}}{{d}^{2}} \left\lceil \frac{40}{{r}^{2}} \right\rceil
\leq
\frac{{r}^{2}}{d}.
\end{equation*}
Therefore
\begin{equation*}
\begin{split}
& \cos{\frac{2 \pi}{d}}
\geq
1 - \frac{1}{2} {\left( \frac{2 \pi}{d} \right)}^{2}
>
1 - \frac{{r}^{2}}{2 d},
\\
& \frac{2 \pi}{d}
<
\mathrm{acos} \left( 1 - \frac{{r}^{2}}{2 d} \right),
\\
& \frac{d}{2 \pi} \mathrm{acos} \left( 1 - \frac{{r}^{2}}{2 d} \right)
>
1.
\end{split}
\end{equation*}
This inequality leaves room for the opposite conclusion, but the second condition comes to the rescue
\begin{equation*}
\frac{{\left\| U - \mathbb{1} \right\|}^{2}}{2 d}
<
2 \sin^2{\frac{\pi}{d}}
=
1 - \cos{\frac{2 \pi}{d}}.
\end{equation*}
Therefore
\begin{equation*}
\begin{split}
& \cos{\frac{2 \pi}{d}}
<
1 - \frac{{\left\| U - \mathbb{1} \right\|}^{2}}{2 d},
\\
& \frac{2 \pi}{d}
>
\mathrm{acos} \left( 1 - \frac{{\left\| U - \mathbb{1} \right\|}^{2}}{2 d} \right),
\\
& \frac{d}{2 \pi} \mathrm{acos} \left( 1 - \frac{{\left\| U - \mathbb{1} \right\|}^{2}}{2 d} \right)
<
1.
\end{split}
\end{equation*}
\end{proof}
\begin{remark}
For $d < \left\lceil \frac{40}{{r}^{2}} \right\rceil$ the condition $\left\| U - \mathbb{1} \right\| < 2 \sqrt{d} \sin{\frac{\pi}{d}}$ might still be needed (it strongly depends on $r$).
\end{remark}
The following is the result we wanted.
\begin{theorem}
Let $r \in \left.\left( 0,1 \right.\right]$ and $U \in SU \left( d \right)$ such that $\left\| U - \mathbb{1} \right\| < r$.
\begin{equation*}
d \geq \left\lceil \frac{40}{{r}^{2}} \right\rceil \quad \text{ and } \quad \left\| U - \mathbb{1} \right\| < 2 \sqrt{d} \sin{\frac{\pi}{d}}
\implies
\log{\left( U \right)} \in \mathfrak{su} \left( d \right).
\end{equation*}
\end{theorem}
\begin{remark}
The condition $\left\| U - \mathbb{1} \right\| < 2 \sqrt{d} \sin{\frac{\pi}{d}}$ in the previous corollary is tight.
Consider indeed the following matrix.
\begin{equation*}
U = 
\begin{pmatrix}
{e}^{\frac{2 \pi i}{d}} &        & 0                       \\
                        & \ddots &                         \\
0                       &        & {e}^{\frac{2 \pi i}{d}}
\end{pmatrix},
\end{equation*}
where $d \geq \left\lceil \frac{40}{{r}^{2}} \right\rceil$.
We see that $\left\| U - \mathbb{1} \right\| = 2 \sqrt{d} \sin{\frac{\pi}{d}} < r$ but $\log{\left( U \right)} \notin \mathfrak{su} \left( d \right)$.
\end{remark}
\end{document}